\documentclass[12pt,a4paper]{article}
\usepackage[utf8]{inputenc}
\usepackage[english]{babel}
\usepackage{amsmath}
\usepackage{amsfonts}
\usepackage{amssymb}
\usepackage[margin=1in]{geometry}
\usepackage{amsthm}
\usepackage{hyperref}
\usepackage[numbers]{natbib}

\newtheorem{as}{Assumption}[section]
\newtheorem{df}{Definition}[section]
\newtheorem{lem}{Lemma}[section]
\newtheorem{prp}{Proposition}[section]
\newtheorem{thm}{Theorem}[section]
\newtheorem{cor}{Corollary}[section]
\newtheorem{rem}{Remark}[section]
\newtheorem{ex}{Example}[section]
\numberwithin{equation}{section}

\title{Term Structure Modeling under Volatility Uncertainty}
\author{Julian H\"olzermann\footnote{Center for Mathematical Economics, Bielefeld University. Email: julian.hoelzermann@uni-bielefeld.de. The author thanks Frank Riedel for valuable advice, Tolulope Fadina and Hanwu Li for fruitful conversations, and Qian Lin, Wolfgang Runggaldier, and the participants of the workshop ``New Frontiers in Stochastics for Economics and Finance" in Siena and the ``12th European Summer School in Financial Mathematics" in Padova for interesting remarks. The author gratefully acknowledges financial support by the German Research Foundation (Deutsche Forschungsgemeinschaft) via Collaborative Research Center 1283.}}

\begin{document}

\maketitle

\begin{abstract}
\noindent In this paper, we study term structure movements in the spirit of Heath, Jarrow, and Morton [\textit{Econometrica 60}(1), 77-105] under volatility uncertainty. We model the instantaneous forward rate as a diffusion process driven by a $G$-Brownian motion. The $G$-Brownian motion represents the uncertainty about the volatility. Within this framework, we derive a sufficient condition for the absence of arbitrage, known as the drift condition. In contrast to the traditional model, the drift condition consists of several equations and several market prices, termed market price of risk and market prices of uncertainty, respectively. The drift condition is still consistent with the classical one if there is no volatility uncertainty. Similar to the traditional model, the risk-neutral dynamics of the forward rate are completely determined by its diffusion term. The drift condition allows to construct arbitrage-free term structure models that are completely robust with respect to the volatility. In particular, we obtain robust versions of classical term structure models.
\end{abstract}


\section{Introduction}
Most approaches to volatility modeling in mathematical finance are subject to model uncertainty, termed \textit{volatility uncertainty}, which can be tamed by making models robust with respect to the volatility. The standard approach to overcome the stylized facts about the volatility in traditional models of mathematical finance is to model the volatility by a stochastic process. The process is typically chosen such that the volatility shares the characteristics of the historical volatility and the resulting option prices are consistent with the current prices observed on the market. However, there are many choices performing this task and it is doubtful to assume that a model specification that is consistent with the past and the present is still valid in the future, since the market environment can change drastically. As the volatility affects the probabilistic law of the underlying, this approach leads to model uncertainty, that is, the uncertainty about using the correct probability measure. To solve this issue, models in robust finance work in the presence of a family of probability measures, consisting of several beliefs about the volatility, which results in models that are robust with respect to the volatility.
\par In the present paper, we study term structure movements in the spirit of \citet*{heathjarrowmorton1992} (HJM) under volatility uncertainty. As in the classical HJM framework, we model the behavior of the instantaneous forward rate as a diffusion process. The forward rate determines all quantities on the related bond market. We represent the uncertainty about the volatility by a family of probability measures, called \textit{set of beliefs}, where each measure represents a different belief about the volatility. Such an approach naturally leads to a sublinear expectation and a $G$-Brownian motion, which was invented by \citet*{peng2019}. As a consequence, the forward rate dynamics are driven by a $G$-Brownian motion in the presence of volatility uncertainty. In contrast to the classical HJM model, the forward rate has uncertain drift terms in addition to the classical (certain) drift term, since the quadratic covariations of a $G$-Brownian motion are uncertain processes. Despite the differences, the present model is still consistent with the classical HJM model. We impose some assumptions on the coefficients of the forward rate dynamics in order to get a sufficient degree of regularity.
\par Similar to the traditional HJM model, the main result of the present work is a drift condition which implies that the related bond market is arbitrage-free. The traditional HJM drift condition relates the absence of arbitrage to the existence of a market price of risk and shows that the risk-neutral dynamics of the forward rate are completely characterized by its diffusion coefficient. In order to derive a drift condition in the presence of volatility uncertainty, we set up a suitable market structure for the related bond market in this setting. In contrast to the traditional HJM model, the drift condition in the presence of volatility uncertainty requires the existence of several market prices. We call the additional market prices (in addition to the market price of risk) the \textit{market prices of uncertainty}. As in the traditional HJM model, the risk-neutral dynamics of the forward rate are completely determined by its diffusion term with the addition that the uncertainty of the diffusion term determines the uncertainty of the drift. If the uncertainty about the volatility vanishes, the drift condition reduces to the traditional one. The proof of the main result is based on deriving the dynamics of the discounted bonds and using a Grisanov transformation for $G$-Brownian motion together with some results on $G$-backward stochastic differential equations.
\par The drift condition derived in this paper is a very powerful tool, since it allows to construct arbitrage-free term structure models that are completely robust with respect to the volatility. In the classical case without volatility uncertainty, almost every (arbitrage-free) term structure model corresponds to a specific example in the HJM methodology. Due to the main result of the present work, we are able to obtain arbitrage-free term structure models in the presence of volatility uncertainty by considering specific examples. In particular, we recover robust versions of classical term structure models. The examples include the Ho-Lee term structure, the Hull-White term structure, and the Vasicek term structure. The examples show that the drift of the risk-neutral short rate dynamics and the bond prices, which still have an affine structure, include an additional uncertain factor when there is uncertainty about the volatility. The interesting thing is that, with this procedure, we obtain term structure models that are robust with respect to the volatility as well as its worst-case values, which differs from most works on pricing under volatility uncertainty.
\par In order to make the analysis from above work, we construct a space of admissible integrands for the forward rate dynamics. The forward rate is a diffusion process parameterized by its maturity, which needs to be integrable with respect to its maturity to compute the bond prices. Therefore, the integrands in the forward rate dynamics need to be regular with respect to the maturity apart from being admissible stochastic processes in a diffusion driven by a $G$-Brownian motion. In order to achieve this, we use the space of Bochner integrable functions, that is, Banach space-valued functions which are sufficiently measurable and integrable, where the functions are mapping from the set of maturities into the space of admissible stochastic processes in this particular case. For such functions, we can define the Bochner integral, mapping into the space of admissible stochastic processes. This ensures that the forward rate is integrable with respect to its maturity. Moreover, we derive further necessary results for the HJM model, including a version of Fubini's theorem for stochastic integrals. We give a sufficient condition for functions to be Bochner integrable, which applies to all considered examples.
\par In addition, we provide a sufficient condition ensuring that the discounted bonds and the portfolio value related to the bond market are well-posed. Each discounted bond is an exponential of a diffusion process driven by a $G$-Brownian motion and the portfolio value consists of integrals with respect to the discounted bonds, respectively. First, we have to make sure that the discounted bonds are well-posed and, second, we need to ensure that the dynamics of each discounted bond are sufficiently regular to imply that the portfolio value is well-posed. For this purpose, we use a condition similar to $G$-Novikov's condition of \citet*{osuka2013} to obtain the desired regularity. As in the classical case, the advantage of such a condition is that it can be easily verified compared to other conditions implying that the exponential of an It\^o diffusion is integrable. One can verify that all examples considered in this article satisfy this condition.
\par The present work is the first in the literature on robust finance that provides a general HJM framework for arbitrage-free term structure modeling in the presence of volatility uncertainty. In general, the literature on interest rates and model uncertainty is relatively sparse. Early contributions in this regard are due to \citet*{avellanedalewicki1996} and \citet*{epsteinwilmott1999}, relying on intuitive arguments rather than a rigorous formulation. A more recent treatment of interest rates in conjunction with model uncertainty appears in the works of \citet*{elkarouiravanelli2009} and \citet*{linriedel2021}, who do not focus on the term structure of interest rates and arbitrage. Further related articles of \citet*{biaginizhang2019} and \citet*{fadinaschmidt2019} deal with credit risk and model uncertainty. Most related are the work of \citet*{fadinaneufeldschmidt2019} and a previous work \citep*{holzermann2021}. \citet*{fadinaneufeldschmidt2019} study affine processes under parameter uncertainty and corresponding interest rate models. However, the results are restricted to short rate models and a superhedging argument for the pricing of contracts, which does not apply to the term structure of interest rates, since bonds are the fundamentals of fixed income markets and cannot be hedged. The previous article \citep*{holzermann2021} specifically focuses on obtaining an arbitrage-free term structure in the presence of volatility uncertainty in a short rate model \`a la Hull-White. The present paper extends the latter results to a general HJM setting.
\par While there are various mathematical approaches to analyze volatility uncertainty, the present paper mainly relies on the calculus of $G$-Brownian motion. Two classical approaches to represent volatility uncertainty are the ones of \citet*{denismartini2006} and \citet*{peng2019}. The approach of \citet*{denismartini2006} works in a probabilistic setting and uses capacity theory. On the other hand, \citet*{peng2019} introduced the calculus of $G$-Brownian motion. In contrast to the first approach, the calculus of $G$-Brownian motion is motivated by nonlinear partial differential equations. However, it holds a duality between both approaches, which was shown by \citet*{denishupeng2011}. Apart from that, one can use an approach based on aggregation \citep*{sonertouzizhang2011b} or a pathwise approach \citep*{contperkowski2019}. In this paper, we start from a probabilistic setting similar to the one of \citet*{denismartini2006}, since it is a natural approach to represent volatility uncertainty from an economic point of view, and use the results of \citet*{denishupeng2011} to acquire all results from the calculus of $G$-Brownian motion. The calculus of $G$-Brownian motion is the main pillar for the mathematical analysis, since the literature on $G$-Brownian motion is very rich and equips us with a lot of mathematical tools.
\par The paper is organized as follows. Section \ref{term structure movements} introduces the forward rate, determining all quantities on the related bond market, and the framework representing the uncertainty about the volatility. In Section \ref{arbitrage-free forward rate dynamics}, we set up a market structure for the related bond market and derive the drift condition, ensuring the absence of arbitrage. In Section \ref{robust versions of classical term structures}, we study examples, including the Ho-Lee term structure, the Hull-White term structure, and the Vasicek term structure, and discuss their implications. Section \ref{conclusion} gives a conclusion. In Section \ref{admissible integrands for the forward rate} of the appendix, we construct the space of admissible integrands for the forward rate dynamics and derive related results. Section \ref{regularity of the discounted bonds} in the appendix provides a sufficient condition for the discounted bonds to be well-posed.

\section{Term Structure Movements}\label{term structure movements}
In the traditional HJM framework, without volatility uncertainty, term structure movements are driven by a standard Brownian motion. That means, we consider the canonical process $B=(B_t^1,...,B_t^d)_{t\geq0}$, for $d\in\mathbb{N}$, on the Wiener space $(\Omega,\mathcal{F},P_0)$ and the filtration $\mathbb{F}=(\mathcal{F}_t)_{t\geq0}$, which is generated by $B$ and completed by all $P_0$-null sets. Then the canonical process $B$ is a $d$-dimensional standard Brownian motion under $P_0$. For $T\leq\tau$, where $\tau<\infty$ is a fixed terminal time, we denote the forward rate with maturity $T$ at time $t$ by $f_t(T)$ for $t\leq T$. For all $T$, the dynamics  of the forward rate process $f(T)=(f_t(T))_{0\leq t\leq T}$ are given by
\begin{align*}
f_t(T)=f_0(T)+\int_0^t\alpha_u(T)du+\sum_{i=1}^d\int_0^t\beta_u^i(T)dB_u^i
\end{align*}
for some initial (observable) forward curve $f_0:[0,\tau]\rightarrow\mathbb{R}$, which is integrable, and sufficiently regular processes $\alpha(T)=(\alpha_t(T))_{0\leq t\leq\tau}$ and $\beta(T)=(\beta_t^1(T),...,\beta_t^d(T))_{0\leq t\leq\tau}$.
\par The forward rate determines all remaining quantities on the bond market. The market offers zero-coupon bonds for all maturities, which are discounted by the money-market account. The bond price process $P(T)=(P_t(T))_{0\leq t\leq T}$ is defined by
\begin{align*}
P_t(T):=\exp\Big(-\int_t^Tf_t(s)ds\Big)
\end{align*}
for all $T\leq\tau$ and the money-market account $M=(M_t)_{0\leq t\leq\tau}$ is defined by
\begin{align*}
M_t:=\exp\Big(\int_0^tr_sds\Big),
\end{align*}
where the short rate $r=(r_t)_{0\leq t\leq\tau}$ is defined by $r_t:=f_t(t)$. We use the money-market account as a num\'eraire, i.e., we focus on the discounted bonds $\tilde{P}(T)=(\tilde{P}_t(T))_{0\leq t\leq T}$ for all $T\leq\tau$, which are defined by
\begin{align*}
\tilde{P}_t(T):=M_t^{-1}P_t(T).
\end{align*}
\par In the presence of volatility uncertainty, we consider a family of probability measures, called \textit{set of beliefs}, where each measure represents a different belief about the volatility. In order to construct the set of beliefs, we consider all scenarios for the volatility which stay in a certain state space. That is, we consider all $\Sigma$-valued, $\mathbb{F}$-adapted processes $\sigma=(\sigma_t)_{t\geq0}$, where $\Sigma$ is a bounded and closed subset of $\mathbb{R}^{d\times d}$. For each such process $\sigma$, we define the process $B^\sigma=(B_t^\sigma)_{t\geq0}$ by
\begin{align*}
B_t^\sigma:=\int_0^t\sigma_udB_u
\end{align*}
and the measure $P^\sigma$ to be the law of the process $B^\sigma$, that is,
\begin{align*}
P^\sigma:=P_0\circ(B^\sigma)^{-1}.
\end{align*}
The set of beliefs is the collection of all such measures, denoted by $\mathcal{P}$. For such a set of measures we define the sublinear expectation $\hat{\mathbb{E}}$ by
\begin{align*}
\hat{\mathbb{E}}[\xi]:=\sup_{P\in\mathcal{P}}\mathbb{E}_P[\xi]
\end{align*}
for all measurable random variables $\xi$ such that $\mathbb{E}_P[\xi]$ exists for all $P\in\mathcal{P}$. One can interprete $\hat{\mathbb{E}}$ as a worst-case measure or as a risk measure.
\par The set of beliefs naturally leads to a $G$-expectation and a $G$-Brownian motion. By the results from \citet*[Theorem 54]{denishupeng2011}, we know that $\hat{\mathbb{E}}$ corresponds to the $G$-expectation on $L_G^1(\Omega)$ and the canonical process $B$ is a $d$-dimensional $G$-Brownian motion under $\hat{\mathbb{E}}$. The $G$-expectation is defined by a nonlinear partial differential equation. The letter $G$ refers to the nonlinear generator $G:\mathbb{S}^d\rightarrow\mathbb{R}$, which is given by
\begin{align*}
G(A)=\tfrac{1}{2}\sup_{\sigma\in\Sigma}\text{tr}(\sigma\sigma'A),
\end{align*}
where $\mathbb{S}^d$ denotes the space of all symmetric $d\times d$ matrices and $\cdot'$ denotes the transpose of a matrix. We assume that $G$ is non-degenerate, i.e, there exists a constant $C>0$ such that $G(A)-G(B)\geq C\text{tr}(A-B)$ for all $A\geq B$. The space $L_G^1(\Omega)$ denotes the space of admissible random variables in the calculus of $G$-Brownian motion, on which we identify random variables that are equal \textit{quasi-surely}, that is, $P$-almost surely for all $P\in\mathcal{P}$. For further insights regarding the calculus of $G$-Brownian motion, the reader may refer to the book of \citet*{peng2019}.
\begin{rem}
As mentioned in the introduction, there are several approaches to model volatility uncertainty and a lot of extensions. In particular, there are extensions to spaces greater than $L_G^1(\Omega)$ and the other related spaces from this calculus. Yet we stick to the classical spaces to use all of the results from the literature on $G$-Brownian motion.
\end{rem}
\par As a consequence, term structure movements are driven by a $G$-Brownian motion in the presence of volatility uncertainty. That means, for all $T$, the forward rate dynamics are now given by
\begin{align*}
f_t(T)=f_0(T)+\int_0^t\alpha_u(T)du+\sum_{i=1}^d\int_0^t\beta_u^i(T)dB_u^i+\sum_{i,j=1}^d\int_0^t\gamma_u^{i,j}(T)d\langle B^i,B^j\rangle_u
\end{align*}
for some initial forward curve $f_0:[0,\tau]\rightarrow\mathbb{R}$, which is integrable, and functions $\alpha,\gamma^{i,j}:[0,\tau]\rightarrow M_G^1(0,\tau)$ and $\beta^i:[0,\tau]\rightarrow M_G^2(0,\tau)$. The space $M_G^p(0,\tau)$ consists of stochastic processes which are admissible integrands in the definition of all stochastic integrals related to a $G$-Brownian motion. Hence, the forward rate and the short rate are well-defined in the sense that, for all $T$, $f_t(T),r_t\in L_G^1(\Omega_t)$ for all $t$.
\par In contrast to the traditional HJM model, the forward rate has additional, uncertain drift terms when there is volatility uncertainty. The additional drift terms of the forward rate are uncertain due to the fact that they depend on the quadratic covariation processes of the $G$-Brownian motion. In the presence of volatility uncertainty, the quadratic covariations of the driving process are uncertain processes, since they differ among the measures in the set of beliefs. Moreover, it can be shown that the additional drift terms cannot be included in the first drift term \citep*[Corollary 3.3]{song2013}. Hence, we have to add them to the forward rate dynamics instead of including them in the first drift term. In this way, we can distinguish between the part of the drift that is driven by uncertainty and the part that is not.
\par When there is no volatility uncertainty, the model corresponds to a classical HJM model. If we drop the uncertainty about the volatility, then $B$ becomes a standard Brownian motion and its quadratic covariation processes are no longer uncertain. That is, if $\Sigma=\{I_d\}$, where $I_d$ denotes the identity matrix, then, for all $i$, $\langle B^i,B^i\rangle_t=t$ and $\langle B^i,B^j\rangle_t=0$ for all $j\neq i$. In that case, the forward rate dynamics are given by
\begin{align*}
f_t(T)=f_0(T)+\int_0^t\Big(\alpha_u(T)+\sum_{i=1}^d\gamma_u^{i,i}(T)\Big)du+\sum_{i=1}^d\int_0^t\beta_u^i(T)dB_u^i
\end{align*}
for all $T$, that is, the model corresponds to a classical HJM model in which the drift is given by the sum of $\alpha$ and $\sum_{i=1}^d\gamma^{i,i}$.
\par We henceforth impose the following two regularity assumptions. The first assumption ensures that the forward rate and the short rate are integrable and that all succeeding computations are feasible.
\begin{as}\label{assumption on alpha, beta, and gamma}
There exists a $p>1$ such that $\alpha,\gamma^{i,j}\in\tilde{M}_G^p(0,\tau)$ and $\beta^i\in\tilde{M}_G^{2p}(0,\tau)$ for all $i,j$.
\end{as}
\noindent The space $\tilde{M}_G^p(0,\tau)$, which we construct in Appendix \ref{admissible integrands for the forward rate}, consists of all functions mapping from $[0,\tau]$ into $M_G^p(0,\tau)$ which are strongly measurable and whose norm on $M_G^p(0,\tau)$ is integrable. A function is called \textit{strongly measurable} if it is Borel measurable and its image is separable. For example, we know that all continuous functions mapping from $[0,\tau]$ into $M_G^p(0,\tau)$ belong to $\tilde{M}_G^p(0,\tau)$ by Proposition \ref{sufficient condition}. This implies that all examples in Section \ref{robust versions of classical term structures} satisfy Assumption \ref{assumption on alpha, beta, and gamma}, as Example \ref{continuous functions are in M tilde} shows. By Proposition \ref{result from bochner integration}, Assumption \ref{assumption on alpha, beta, and gamma} implies that the forward rate and, by Proposition \ref{complex integrals}, the short rate are integrable. The second assumption ensures that the discounted bonds and the portfolio value are sufficiently regular.
\begin{as}\label{assumption on a, b, and c}
There exist $\tilde{p}>p^*$ and $\tilde{q}>2$, where $p^*:=\frac{2pq}{p-q}$ for some $q\in(1,p)$, such that, for all $T\leq\tau$, it holds
\begin{align*}
\hat{\mathbb{E}}\Big[\int_0^T\exp\Big(\tfrac{\tilde{p}\tilde{q}}{\tilde{q}-2}\Big(\int_0^ta_u(T)du+\sum_{i,j=1}^d\int_0^tc_u^{i,j}(T)d\langle B^i,B^j\rangle_u\Big)\Big)dt\Big]<{}&\infty,
\\\hat{\mathbb{E}}\Big[\int_0^T\exp\Big(\tfrac{1}{2}(\tilde{p}\tilde{q})^2\sum_{i,j=1}^d\int_0^tb_u^i(T)b_u^j(T)d\langle B^i,B^j\rangle_u\Big)dt\Big]<{}&\infty.
\end{align*}
\end{as} 
\noindent The processes $a(T)=(a_t(T))_{0\leq t\leq\tau}$, $b^i(T)=(b_t^i(T))_{0\leq t\leq\tau}$, and $c^{i,j}(T)=(c_t^{i,j}(T))_{0\leq t\leq\tau}$, for $T\leq\tau$, are defined by
\begin{align*}
a_t(T):=\int_t^T\alpha_t(s)ds,\quad b_t^i(T):=\int_t^T\beta_t^i(s)ds,\quad c_t^{i,j}(T):=\int_t^T\gamma_t^{i,j}(s)ds
\end{align*}
for all $i,j$, for which we have $a(T),c^{i,j}(T)\in M_G^p(0,\tau)$ and $b^i(T)\in M_G^{2p}(0,\tau)$ by Assumption \ref{assumption on alpha, beta, and gamma} and Proposition \ref{complex integrals}. One can easily verify that all examples in Section \ref{robust versions of classical term structures} satisfy Assumption \ref{assumption on a, b, and c}. Assumption \ref{assumption on a, b, and c} is similar to $G$-Novikov's condition from \citet*{osuka2013} and implies that, for every maturity, the discounted bond price is in $L_G^1(\Omega_t)$ at each time $t$. Moreover, Assumption \ref{assumption on a, b, and c} ensures that the dynamics of the discounted bonds are regular enough to imply that the portfolio value, which is defined below, is well-posed. We show both of these implications in Appendix \ref{regularity of the discounted bonds}, which requires Lemma \ref{dynamics of the logarithm of the discounted bond} from below.

\section{Arbitrage-Free Forward Rate Dynamics}\label{arbitrage-free forward rate dynamics}
In the traditional HJM model, the absence of arbitrage on the related bond market is ensured by the HJM drift condition, which assumes the existence of a market price of risk and characterizes the drift of the forward rate in terms of its diffusion coefficient. More precisely, the market is arbitrage-free if there exists a suitable process $\lambda=(\lambda_t^1,...,\lambda_t^d)_{0\leq t\leq\tau}$ such that, for all $T$,
\begin{align*}
\alpha(T)-\beta(T)b(T)'+\beta(T)\lambda'=0,
\end{align*}
where $b(T)=(b^1(T),...,b^d(T))$. The process $\lambda$ is termed \textit{market price of risk}, since it erases the drift of each discounted bond under an equivalent probability measure, called \textit{risk-neutral measure}, to make it a martingale. Then the forward rate dynamics under the risk-neutral measure are completely determined by its diffusion coefficient, that is,
\begin{align*}
f_t(T)=f_0(T)+\int_0^t\beta_u(T)b_u(T)'du+\sum_{i=1}^d\int_0^t\beta_u^i(T)d\bar{B}_u^i,
\end{align*}
where $\bar{B}=(\bar{B}_t^1,...,\bar{B}_t^d)_{0\leq t\leq\tau}$ is a Brownian motion under the risk-neutral measure. This fact is of practical importance, since there is no need to specify the drift term $\alpha$ or the market price of risk $\lambda$.
\par In order to derive a drift condition in the presence of volatility uncertainty, we first define admissible market strategies and a suitable notion of arbitrage. The agents in the market are allowed to select a finite number of bonds they want to trade. The corresponding portfolio value is determined by the gains from trade, i.e., we restrict to self-financing strategies.
\begin{df}\label{market strategy}
An admissible market strategy $(\pi,T)$ is a couple consisting of a bounded process $\pi=(\pi_t^1,...,\pi_t^n)_{0\leq t\leq\tau}$ in $M_G^2(0,\tau;\mathbb{R}^n)$ and a vector $T=(T_1,...,T_n)\in[0,\tau]^n$ for some $n\in\mathbb{N}$. The corresponding portfolio value at terminal time is defined by
\begin{align*}
\tilde v_\tau(\pi,T):=\sum_{i=1}^n\int_0^{T_i}\pi_t^id\tilde{P}_t(T_i).
\end{align*}
\end{df}
\noindent The portfolio value is well-posed, since the dynamics of $\tilde{P}(T)$, which are derived in Proposition \ref{dynamics of the discounted bond} below, are sufficiently regular for each $T$ by Assumption \ref{assumption on a, b, and c} and Proposition \ref{dynamics of the discounted bond are regular}. In addition, we use the quasi-sure notion of arbitrage, which corresponds to the one frequently used in the literature on model uncertainty \citep*{biaginibouchardkardarasnutz2017,bouchardnutz2015}.
\begin{df}\label{arbitrage}
An admissible market strategy $(\pi,T)$ is called arbitrage if
\begin{align*}
\tilde{v}_\tau(\pi,T)\geq0\quad\text{quasi-surely},\quad P\big(\tilde{v}_\tau(\pi,T)>0\big)>0\quad\text{for at least one }P\in\mathcal{P}.
\end{align*}
Moreover, we say that the bond market is arbitrage-free if there is no arbitrage.
\end{df}
\begin{rem}
It is possible to generalize the notion of trading strategies and the concept of arbitrage. The notion of trading strategies can be generalized by allowing for measure-valued trading strategies \citep*{bjorkdimasikabanovrunggaldier1997} or by using methods from large financial markets \citep*{kleinschmidtteichmann2016}. There are also other no-arbitrage concepts related to bond markets, which are based on the theory of large financial markets \citep*{cuchierokleinteichmann2016}. We stick to the definitions from above, since such a generalization is not the objective of the present paper.
\end{rem}
\par In the presence of volatility uncertainty, the absence of arbitrage requires the existence of several market prices. In that case, there is a sublinear expectation under which the discounted bonds are symmetric $G$-martingales, ruling out arbitrage opportunities. Moreover, the drift condition characterizes the dynamics of the forward rate.
\begin{thm}\label{no-arbitrage if the drift condition is satisfied}
The bond market is arbitrage-free if, for some $p>1$, there exist bounded processes $\kappa=(\kappa_t^1,...,\kappa_t^d)_{0\leq t\leq\tau}$ and $\lambda^{i,j}=(\lambda_t^{i,j,1},...,\lambda_t^{i,j,d})_{0\leq t\leq\tau}$ in $M_G^p(0,\tau;\mathbb{R}^d)$ such that
\begin{align}\label{drift condition}
\begin{split}
\alpha(T)+\beta(T)\kappa'={}&0,
\\\gamma^{i,j}(T)-\tfrac{1}{2}\big(\beta^i(T)b^j(T)+b^i(T)\beta^j(T)\big)+\beta(T)(\lambda^{i,j})'={}&0
\end{split}
\end{align}
for almost all $T$ for all $i,j$. In particular, then there exists a sublinear expectation $\bar{\mathbb{E}}$ under which $\tilde{P}(T)$, for all $T$, is a symmetric $G$-martingale and, for almost all $T$,
\begin{align*}
f_t(T)=f_0(T)+\sum_{i=1}^d\int_0^t\beta_u^i(T)d\bar{B}_u^i+\sum_{i,j=1}^d\int_0^t\tfrac{1}{2}\big(\beta_u^i(T)b_u^j(T)+b_u^i(T)\beta_u^j(T)\big)d\langle\bar{B}^i,\bar{B}^j\rangle_u,
\end{align*}
where $\bar{B}=(\bar{B}_t^1,...,\bar{B}_t^d)_{0\leq t\leq\tau}$ is a $G$-Brownian motion under $\bar{\mathbb{E}}$.
\end{thm}
\par The additional market prices occuring in the drift condition represent the market prices of uncertainty. In comparison to the classical case without volatility uncertainty, the forward rate and (hence) the discounted bonds have additional drift terms, which are uncertain (as explained in Section \ref{term structure movements}). As a consequence, we need additional market prices in order to make the discounted bonds symmetric $G$-martingales, which ultimately rules out arbitrage. Since the additional market prices relate to the uncertain drift terms of the discounted bonds, they are termed \textit{market prices of uncertainty}.
\par The risk-neutral dynamics of the forward rate are fully characterized by its diffusion term, which, in contrast to the classical HJM model, does not only apply to the coefficients but also to the uncertainty. We call the dynamics of the forward rate under $\bar{\mathbb{E}}$, given in Theorem \ref{no-arbitrage if the drift condition is satisfied}, \textit{risk-neutral dynamics}, since the discounted bonds are symmetric $G$-martingales under $\bar{\mathbb{E}}$. As in the classical HJM model, the diffusion coefficient $\beta$ determines the drift coefficient of the risk-neutral forward rate dynamics. In addition, the uncertain volatility, included in the $G$-Brownian motion $\bar{B}$, determines the uncertainty of the drift, represented by the quadratic covariation processes of $\bar{B}$. That means, arbitrage-free term structure models also exhibit drift uncertainty in the presence of volatility uncertainty.
\par Despite the differences, the drift condition is still consistent with the classical HJM drift condition. If there is no uncertainty about the volatility, that is, if $\Sigma=\{I_d\}$, then the forward rate, for all $T$, satisfies
\begin{align*}
f_t(T)=f_0(T)+\int_0^t\Big(\alpha_u(T)+\sum_{i=1}^d\gamma_u^{i,i}(T)\Big)du+\sum_{i=1}^d\int_0^t\beta_u^i(T)dB_u^i
\end{align*}
as it is described in Section \ref{term structure movements}. The drift condition in Theorem \ref{no-arbitrage if the drift condition is satisfied} implies that
\begin{align*}
\Big(\alpha(T)+\sum_{i=1}^d\gamma^{i,i}(T)\Big)-\beta(T)b(T)'+\beta(T)\Big(\kappa+\sum_{i=1}^d\lambda^{i,i}\Big)'=0
\end{align*}
for almost all $T$. The latter corresponds to the classical HJM drift condition for a market price of risk given by the process $\kappa+\sum_{i=1}^d\lambda^{i,i}$.
\par In order to prove Theorem \ref{no-arbitrage if the drift condition is satisfied}, we first of all derive the dynamics of the discounted bond for each maturity. This is based on the following lemma.
\begin{lem}\label{dynamics of the logarithm of the discounted bond}
For all $T$, the logarithm of $\tilde{P}(T)$ satisfies the dynamics
\begin{align*}
\log\big(\tilde{P}_t(T)\big)=\log\big(\tilde{P}_0(T)\big)-\int_0^ta_u(T)du-\sum_{i=1}^d\int_0^tb_u^i(T)dB_u^i-\sum_{i,j=1}^d\int_0^tc_u^{i,j}(T)d\langle B^i,B^j\rangle_u.
\end{align*}
\end{lem}
\begin{proof}
We obtain the dynamics by applying Fubini's theorem, which can be found in Appendix \ref{admissible integrands for the forward rate}. Inserting the forward rate dynamics and the short rate dynamics into the definition of $\tilde{P}(T)$, the logarithm of $\tilde{P}(T)$ satisfies
\begin{align*}
\log\big(\tilde{P}_t(T)\big)={}&\log\big(\tilde{P}_0(T)\big)-\int_t^T\int_0^t\alpha_u(s)duds-\int_0^t\int_0^s\alpha_u(s)duds
\\&{}-\sum_{i=1}^d\int_t^T\int_0^t\beta_u^i(s)dB_u^ids-\sum_{i=1}^d\int_0^t\int_0^s\beta_u^i(s)dB_u^ids
\\&{}-\sum_{i,j=1}^d\int_t^T\int_0^t\gamma_u^{i,j}(s)d\langle B^i,B^j\rangle_uds-\sum_{i,j=1}^d\int_0^t\int_0^s\gamma_u^{i,j}(s)d\langle B^i,B^j\rangle_uds
\end{align*}
for all $T$. Then an application of Corollary \ref{fubini} yields, for all $i,j$,
\begin{align*}
\int_t^T\int_0^t\alpha_u(s)duds+\int_0^t\int_0^s\alpha_u(s)duds={}&\int_0^ta_u(T)du,
\\\int_t^T\int_0^t\beta_u^i(s)dB_u^ids+\int_0^t\int_0^s\beta_u^i(s)dB_u^ids={}&\int_0^tb_u^i(T)dB_u^i,
\\\int_t^T\int_0^t\gamma_u^{i,j}(s)d\langle B^i,B^j\rangle_uds+\int_0^t\int_0^s\gamma_u^{i,j}(s)d\langle B^i,B^j\rangle_uds={}&\int_0^tc_u^{i,j}(T)d\langle B^i,B^j\rangle_u
\end{align*}
for all $T$, which proves the assertion.
\end{proof}
\noindent Since, for all $i,j$, the processes $a(T)$, $b^i(T)$, and $c^{i,j}(T)$ are sufficiently regular for each $T$, we can use It\^o's formula for $G$-Brownian motion from \citet*{lipeng2011} to derive the dynamics of the discounted bond for each $T$.
\begin{prp}\label{dynamics of the discounted bond}
For all $T$, the discounted bond $\tilde{P}(T)$ satisfies the dynamics
\begin{align*}
\tilde{P}_t(T)={}&\tilde{P}_0(T)-\int_0^ta_u(T)\tilde{P}_u(T)du-\sum_{i=1}^d\int_0^tb_u^i(T)\tilde{P}_u(T)dB_u^i
\\&{}-\sum_{i,j=1}^d\int_0^t\big(c_u^{i,j}(T)-\tfrac{1}{2}b_u^i(T)b_u^j(T)\big)\tilde{P}_u(T)d\langle B^i,B^j\rangle_u.
\end{align*}
\end{prp}
\begin{proof}
The assertion follows by Lemma \ref{dynamics of the logarithm of the discounted bond} and an application of It\^o's formula for $G$-Brownian motion \citep*[Theorem 5.4]{lipeng2011}. We are able to apply the latter, since $a(T),c^{i,j}(T)\in M_G^1(0,T)$ and $b^i(T)\in M_G^2(0,T)$ for all $i,j$ by the contruction of the Bochner integral in Appendix \ref{admissible integrands for the forward rate} and Proposition \ref{complex integrals}.
\end{proof}
\par We prove Theorem \ref{no-arbitrage if the drift condition is satisfied} by using results on $G$-backward stochastic differential equations of \citet*{hujipengsong2014}, including a Girsanov transformation for $G$-Brownian motion. Details regarding $G$-backward stochastic differential equations are stated in the paper of \citet*{hujipengsong2014}.
\begin{proof}[Proof of Theorem \ref{no-arbitrage if the drift condition is satisfied}]
First, we rewrite the dynamics of the forward rate and the discounted bond for each maturity by using the Girsanov transformation for $G$-Brownian motion from \citet*{hujipengsong2014}. For this purpose, we consider the following sublinear expectation. For $\xi\in L_G^p(\Omega_\tau)$ with $p>1$, we define the sublinear expectation $\bar{\mathbb{E}}$ by $\bar{\mathbb{E}}_t[\xi]:=Y_t^\xi$, where $Y^\xi=(Y_t^\xi)_{0\leq t\leq\tau}$ solves the $G$-backward stochastic differential equation
\begin{align*}
Y_t^\xi=\xi+\int_t^\tau\kappa_uZ_u'du+\sum_{i,j=1}^d\int_t^\tau\lambda_u^{i,j}Z_u'd\langle B^i,B^j\rangle_u-\sum_{i=1}^d\int_t^\tau Z_u^idB_u^i-(K_\tau-K_t).
\end{align*}
Then $\bar{\mathbb{E}}$ is a time consistent sublinear expectation \citep*[Theorem 5.1]{hujipengsong2014} and the process $\bar{B}=(\bar{B}_t^1,...,\bar{B}_t^d)_{0\leq t\leq\tau}$, defined by
\begin{align*}
\bar{B}_t:=B_t-\int_0^t\kappa_udu-\sum_{i,j=1}^d\int_0^t\lambda_u^{i,j}d\langle B^i,B^j\rangle_u,
\end{align*}
is a $G$-Brownian motion under $\bar{\mathbb{E}}$ \citep*[Theorems 5.2, 5.4]{hujipengsong2014}. The quadratic covariations of $B$ and $\bar{B}$, respectively, are the same, since the drift terms of $\bar{B}$ are of bounded variation. Hence, for each $T$, we can rewrite the dynamics of the forward rate as
\begin{align*}
f_t(T)={}&f_0(T)+\int_0^t\big(\alpha_u(T)+\beta_u(T)\kappa_u'\big)du+\sum_{i=1}^d\int_0^t\beta_u^i(T)d\bar{B}_u^i
\\&{}+\sum_{i,j=1}^d\int_0^t\big(\gamma_u^{i,j}(T)+\beta_u(T)(\lambda_u^{i,j})'\big)d\langle\bar{B}^i,\bar{B}^j\rangle_u
\end{align*}
and the dynamics of the discounted bond as
\begin{align*}
\tilde{P}_t(T)={}&\tilde{P}_0(T)-\int_0^t\big(a_u(T)+b_u(T)\kappa_u'\big)\tilde{P}_u(T)du-\sum_{i=1}^d\int_0^tb_u^i(T)\tilde{P}_u(T)d\bar{B}_u^i
\\&{}-\sum_{i,j=1}^d\int_0^t\big(c_u^{i,j}(T)-\tfrac{1}{2}b_u^i(T)b_u^j(T)+b_u(T)(\lambda_u^{i,j})'\big)\tilde{P}_u(T)d\langle\bar{B}^i,\bar{B}^j\rangle_u.
\end{align*}
\par Next, we deduce the forward rate dynamics and the dynamics of the discounted bonds under $\bar{\mathbb{E}}$ from the drift condition. By \eqref{drift condition}, for almost all $T$,
\begin{align*}
f_t(T)=f_0(T)+\sum_{i=1}^d\int_0^t\beta_u^i(T)d\bar{B}_u^i+\sum_{i,j=1}^d\int_0^t\tfrac{1}{2}\big(\beta_u^i(T)b_u^j(T)+b_u^i(T)\beta_u^j(T)\big)d\langle\bar{B}^i,\bar{B}^j\rangle_u.
\end{align*}
Additionally, we can integrate the terms in \eqref{drift condition} to get, for all $i,j$,
\begin{align*}
\int_\cdot^T\big(\alpha(s)+\beta(s)\kappa'\big)ds={}&a(T)+b(T)\kappa',
\\\int_\cdot^T\big(\gamma^{i,j}(s)+\beta(s)(\lambda^{i,j})'\big)ds={}&c^{i,j}(T)+b(T)(\lambda^{i,j})',
\\\int_\cdot^T\big(\beta^i(s)b^j(s)+b^i(s)\beta^j(s)\big)ds={}&b^i(T)b^j(T)
\end{align*}
for all $T$. The latter follows from Corollary \ref{product rule for the diffusion coefficient}. Thus, by \eqref{integral inequality}, for all $i,j$,
\begin{align*}
a(T)+b(T)\kappa'={}&0,
\\c^{i,j}(T)-\tfrac{1}{2}b^i(T)b^j(T)+b(T)(\lambda^{i,j})'={}&0
\end{align*}
for all $T$, which implies
\begin{align*}
\tilde{P}_t(T)=\tilde{P}_0(T)-\sum_{i=1}^d\int_0^tb_u^i(T)\tilde{P}_u(T)d\bar{B}_u^i.
\end{align*}
\par In the end, we conclude that the market is arbitrage-free, since the discounted bonds are symmetric $G$-martingales. From Assumption \ref{assumption on a, b, and c} and Proposition \ref{dynamics of the discounted bond are regular}, we deduce that, for each $T$, $\tilde{P}_t(T)\in L_G^2(\Omega_t)$ for all $t$. In addition, the dynamics of the discounted bonds from above imply that $\tilde{P}(T)$ is a symmetric $G$-martingale under $\bar{\mathbb{E}}$ for all $T$. Therefore, it can be shown that the bond market is arbitrage-free \citep*[Proposition 5.1]{holzermann2021}. This relies on the fact that $\bar{\mathbb{E}}$ is equivalent to the initial sublinear expectation $\hat{\mathbb{E}}$ in the sense that $\xi=0$ if and only if $\bar{\mathbb{E}}[\vert\xi\vert]=0$ for $\xi\in L_G^p(\Omega_\tau)$ with $p>1$ \citep*[Lemma 5.1]{holzermann2021}.
\end{proof}

\section{Robust Versions of Classical Term Structures}\label{robust versions of classical term structures}
The main reason the HJM methodology is so popular is that essentially every term structure model corresponds to a specific example in the HJM model. One can verify that the forward rate implied by any arbitrage-free term structure satisfies the HJM drift condition for a particular diffusion coefficient. Conversely, the diffusion coefficient fully characterizes the risk-neutral dynamics of the forward rate and the forward rate, in turn, determines all other quantities of the model. Thus, one is able to construct arbitrage-free term structure models by simply specifying the diffusion term of the forward rate.
\par We investigate what kind of term structure models we obtain when we consider specific examples in the present setting. Theorem \ref{no-arbitrage if the drift condition is satisfied} shows that the risk-neutral dynamics of the forward rate are also determined by its diffusion term when there is volatility uncertainty. Therefore, the drift condition derived in the previous section enables us to construct arbitrage-free term structure models in the presence of volatility uncertainty by specifying the diffusion term of the forward rate, which we demonstrate in the succeeding examples. In particular, our aim is to recover robust versions of classical term structure models by considering the corresponding diffusion coefficients in the present framework, respectively.
\par Throughout the section, we impose the following assumptions. First, we consider a one-dimensional $G$-Brownian motion, that is, $d=1$ and $\Sigma=[\underline{\sigma},\overline{\sigma}]$ for $\overline{\sigma}\geq\underline{\sigma}>0$. Second, we suppose that the initial forward curve is differentiable. This is necessary for the derivation of the related short rate dynamics. Third, we assume that the drift condition is satisfied. This assumption ensures that the model is arbitrage-free and allows us to directly compute the risk-neutral dynamics of the forward rate. It should also be noted that the following examples are feasible in the sense that the respective diffusion coefficients satisfy the regularity assumptions from Section \ref{term structure movements}.

\subsection{The Ho-Lee Term Structure}\label{the ho-lee term structure}
If we consider the diffusion coefficient of the forward rate implied by the Ho-Lee term structure, we obtain a robust version of the Ho-Lee model.
\begin{ex}\label{ho-lee term structure and short rate dynamics}
If we define $\beta$ by $\beta_t(T):=1$, then the short rate satisfies
\begin{align*}
r_t=r_0+\int_0^t\big(\partial_uf_0(u)+q_u\big)du+\bar{B}_t
\end{align*}
and the bond prices are of the form
\begin{align*}
P_t(T)=\exp\big(A(t,T)-\tfrac{1}{2}B(t,T)^2q_t-B(t,T)r_t\big),
\end{align*}
where the process $q=(q_t)_{0\leq t\leq\tau}$ is defined by
\begin{align*}
q_t:=\langle\bar{B}\rangle_t
\end{align*}
and the functions $A,B:[0,\tau]\times[0,\tau]\rightarrow\mathbb{R}$ are defined by
\begin{align*}
A(t,T)&{}:=-\int_t^Tf_0(s)ds+B(t,T)f_0(t),
\\B(t,T)&{}:=(T-t).
\end{align*}
\par The risk-neutral short rate dynamics are determined by the risk-neutral forward rate dynamics. According to Theorem \ref{no-arbitrage if the drift condition is satisfied}, the latter are given by
\begin{align*}
f_t(T)=f_0(T)+\bar{B}_t+\int_0^t(T-u)d\langle\bar{B}\rangle_u.
\end{align*}
By the definition of the short rate, we have
\begin{align*}
r_t=f_0(t)+\bar{B}_t+\int_0^t(t-u)d\langle\bar{B}\rangle_u.
\end{align*}
Applying It\^o's formula for $G$-Brownian motion then yields
\begin{align*}
r_t=r_0+\int_0^t\big(\partial_uf_0(u)+q_u\big)du+\bar{B}_t.
\end{align*}
\par The bond prices follow from integrating the risk-neutral forward rate dynamics. We have
\begin{align*}
\int_t^Tf_t(s)ds=\int_t^Tf_0(s)ds+B(t,T)\bar{B}_t+\int_t^T\int_0^t(s-u)d\langle\bar{B}\rangle_uds.
\end{align*}
If we perform some calculations on the last term, we get
\begin{align*}
\int_t^T\int_0^t(s-u)d\langle\bar{B}\rangle_uds=B(t,T)\int_0^t(t-u)d\langle\bar{B}\rangle_u+\tfrac{1}{2}B(t,T)^2\langle\bar{B}\rangle_t.
\end{align*}
Substituting the latter in the previous equation, we obtain
\begin{align*}
\int_t^Tf_t(s)ds=-A(t,T)+\tfrac{1}{2}B(t,T)^2q_t+B(t,T)r_t,
\end{align*}
which yields the bond prices from above.
\end{ex}

\subsection{The Hull-White Term Structure}\label{the hull-white term structure}
If we use the diffusion coefficient of the forward rate implied by the Hull-White term structure, we get a robust version of the Hull-White model.
\begin{ex}\label{hull-white term structure and short rate dynamics}
If we define $\beta$ by $\beta_t(T):=e^{-\theta(T-t)}$ for $\theta>0$, then the short rate dynamics are given by
\begin{align*}
r_t=r_0+\int_0^t\big(\partial_uf_0(u)+\theta f_0(u)+q_u-\theta r_u\big)du+\bar{B}_t
\end{align*}
and the bond prices are of the form
\begin{align*}
P_t(T)=\exp\big(A(t,T)-\tfrac{1}{2}B(t,T)^2q_t-B(t,T)r_t\big),
\end{align*}
where the process $q=(q_t)_{0\leq t\leq\tau}$ is defined by
\begin{align*}
q_t:=\int_0^te^{-2\theta(t-u)}d\langle\bar{B}\rangle_u
\end{align*}
and the functions $A,B:[0,\tau]\times[0,\tau]\rightarrow\mathbb{R}$ are defined by
\begin{align*}
A(t,T)&{}:=-\int_t^Tf_0(s)ds+B(t,T)f_0(t),
\\B(t,T)&{}:=\tfrac{1}{\theta}(1-e^{-\theta(T-t)}).
\end{align*}
\par Again, the risk-neutral short rate dynamics are determined by the risk-neutral forward rate dynamics. By Theorem \ref{no-arbitrage if the drift condition is satisfied}, the latter are given by
\begin{align*}
f_t(T)=f_0(T)+\int_0^te^{-\theta(T-u)}d\bar{B}_u+\int_0^te^{-\theta(T-u)}\tfrac{1}{\theta}(1-e^{-\theta(T-u)})d\langle\bar{B}\rangle_u.
\end{align*}
The definition of the short rate then implies
\begin{align*}
r_t=f_0(t)+\int_0^te^{-\theta(t-u)}d\bar{B}_u+\int_0^te^{-\theta(t-u)}\tfrac{1}{\theta}(1-e^{-\theta(t-u)})d\langle\bar{B}\rangle_u.
\end{align*}
Applying It\^o's formula for $G$-Brownian motion yields
\begin{align*}
r_t=r_0+\int_0^t\big(\partial_uf_0(u)+\theta f_0(u)+q_u-\theta r_u\big)du+\bar{B}_t.
\end{align*}
\par We obtain the bond prices by integrating the risk-neutral dynamics of the forward rate. We have
\begin{align*}
\int_t^Tf_t(s)ds={}&\int_t^Tf_0(s)ds+\int_t^T\int_0^te^{-\theta(s-u)}d\bar{B}_uds
\\&{}+\int_t^T\int_0^te^{-\theta(s-u)}\tfrac{1}{\theta}(1-e^{-\theta(s-u)})d\langle\bar{B}\rangle_uds.
\end{align*}
The first double integral can be written as
\begin{align*}
\int_t^T\int_0^te^{-\theta(s-u)}d\bar{B}_uds=B(t,T)\int_0^te^{-\theta(t-u)}d\bar{B}_u.
\end{align*}
After some calculations, the second double integral becomes
\begin{align*}
\int_t^T\int_0^te^{-\theta(s-u)}\tfrac{1}{\theta}(1-e^{-\theta(s-u)})d\langle\bar{B}\rangle_uds={}&B(t,T)\int_0^te^{-\theta(t-u)}\tfrac{1}{\theta}(1-e^{-\theta(t-u)})d\langle\bar{B}\rangle_u
\\&{}+\tfrac{1}{2}B(t,T)^2\int_0^te^{-2\theta(t-u)}d\langle\bar{B}\rangle_u.
\end{align*}
Thus, we obtain
\begin{align*}
\int_t^Tf_t(s)ds=-A(t,T)+\tfrac{1}{2}B(t,T)^2q_t+B(t,T)r_t,
\end{align*}
which leads to the bond prices given above.
\end{ex}
\begin{rem}
The Hull-White model under volatility uncertainty is also analyzed in a previous work \citep*{holzermann2021}. The paper shows how to obtain an arbitrage-free term structure in the Hull-White model when there is uncertainty about the volatility. In order to achieve this, the structure of the short rate dynamics has to be suitably modified. Here we get exactly the same structure.
\end{rem}

\subsection{The Vasicek Term Structure}\label{the vasicek term structure}
The previous example shows that the Vasicek model needs to be adjusted in order to fit into the HJM methodology when there is volatility uncertainty.
\begin{ex}\label{adjusting the vasicek model}
If we use the same diffusion coefficient as in the previous example, we see that it is not possible to exactly replicate the Vasicek model  in the presence of volatility uncertainty. The forward rates implied by the Vasicek term structure and the Hull-White term structure, respectively, have the same diffusion coefficient. If we define $\beta$ as in Example \ref{hull-white term structure and short rate dynamics}, then the short rate dynamics are given by
\begin{align*}
r_t=r_0+\int_0^t\big(\partial_uf_0(u)+\theta f_0(u)+q_u-\theta r_u\big)du+\bar{B}_t,
\end{align*}
where the process $q$ is defined as in Example \ref{hull-white term structure and short rate dynamics}. In order to obtain the short rate dynamics of the Vasicek model, we need to make sure that, for all $t$,
\begin{align}\label{equation ensuring a constant mean reversion level}
\partial_tf_0(t)+\theta f_0(t)+q_t=\mu
\end{align}
for a constant $\mu>0$, since the mean reversion level of the short rate is constant in the Vasicek model. As equation \eqref{equation ensuring a constant mean reversion level} does not hold for any initial forward curve, the equation imposes a condition on $f_0$ that ensures a constant mean reversion level. If there is no volatility uncertainty, one can check that the initial forward curve of the Vasicek term structure satisfies \eqref{equation ensuring a constant mean reversion level}. In the presence of volatility uncertainty, there is no initial forward curve $f_0$ satisfying \eqref{equation ensuring a constant mean reversion level}, since then the process $q$ is uncertain, i.e., its realization $q_t$ is only known after time $t$, while $f_0$ is observable at inception.
\par We can circumvent the problem by modifying the Vasicek model. Let us suppose that the initial forward curve satisfies
\begin{align*}
\partial_tf_0(t)+\theta f_0(t)=\mu
\end{align*}
for all $t$. That means, $f_0$ solves a simple ordinary differential equation with initial condition $f_0(0)=r_0$, which yields
\begin{align*}
f_0(t)=e^{-\theta t}r_0+\mu B(0,t)
\end{align*}
for all $t$, where the function $B$ is defined as in Example \ref{hull-white term structure and short rate dynamics}. Then the short rate dynamics are given by
\begin{align*}
r_t=r_0+\int_0^t(\mu+q_u-\theta r_u)du+\bar{B}_t
\end{align*}
and, as in Example \ref{hull-white term structure and short rate dynamics}, the bond prices are of the form
\begin{align*}
P_t(T)=\exp\big(A(t,T)-\tfrac{1}{2}B(t,T)^2q_t-B(t,T)r_t\big),
\end{align*}
where the function $A$, for all $t$ and $T$, now satisfies
\begin{align*}
A(t,T)=-\mu\int_t^TB(s,T)ds.
\end{align*}
So, instead of exactly replicating the Vasicek model, we can obtain a version of the Vasicek model in which the mean reversion level of the short rate is adjusted by the process $q$.
\end{ex}
\begin{rem}
The problem mentioned in Example \ref{adjusting the vasicek model} does not occur in Example \ref{hull-white term structure and short rate dynamics} due to the time dependent mean reversion level in the Hull-White model. Since the mean reversion level can be time dependent and possibly uncertain, equation \eqref{equation ensuring a constant mean reversion level} only imposes a condition on the mean reversion level but not on the initial forward curve.
\par In fact, the problematic is related to the ability of term structure models to match arbitrary forward curves observed on the market, since the HJM methodology is based on modeling the forward rate starting from an arbitrary initial forward curve. The Hull-White model (as well as the Ho-Lee model) involves a time dependent parameter and hence, offers enough flexibility to fit the model-implied term structure to any term structure obtained from data. On the other hand, the Vasicek model has only three (constant) parameters, restricting the model to a small class of term structures it can fit. Therefore, one has to impose further assumptions on $f_0$ to reproduce the Vasicek model in the HJM methodology in general, which, however, does not work when there is volatility uncertainty.
\end{rem}

\subsection{Economic Consequences}\label{economic consequences}
The examples show that arbitrage-free term structure models in the presence of volatility uncertainty exhibit an additional uncertain factor. In all examples we consider there is an uncertain process, always denoted by $q$, that enters the short rate dynamics and the bond prices. The process $q$ is uncertain, since it depends on the quadratic variation process. It emerges due to the fact that the risk-neutral forward rate dynamics display drift uncertainty, as Theorem \ref{no-arbitrage if the drift condition is satisfied} shows. That means, the additional factor is required in order to make the model arbitrage-free when the volatility is uncertain. Despite this difference, the short rate dynamics and the bond prices still have an affine form. In fact, the resulting term structure models, except the one in Example \ref{adjusting the vasicek model}, are consistent with the classical ones. If there is no volatility uncertainty, that is, if $\overline{\sigma}=\underline{\sigma}$, then Examples \ref{ho-lee term structure and short rate dynamics} and \ref{hull-white term structure and short rate dynamics} correspond to the traditional Ho-Lee model \citep*[Subsection 5.4.4]{filipovic2009} and the traditional Hull-White model \citep*[Subsections 3.3.1, 3.3.2]{brigomercurio2001}, respectively. In that case, the process $q$ is no longer uncertain. Hence, the process $q$ is actually included in the Ho-Lee model and in the Hull-White model, but it is hidden as there is no volatility uncertainty in the traditional models. The Vasicek model does not include such a factor and thus, it has to be adjusted in order to be arbitrage-free in the presence of volatility, as demonstrated in Example \ref{adjusting the vasicek model}.
\par The term structure models resulting from the examples are completely robust with respect to the volatility. The bond prices in the examples are completely independent of the volatility. They do not even depend on the extreme values of the volatility, $\overline{\sigma}$ and $\underline{\sigma}$, which usually happens in option pricing under volatility uncertainty. Of course, such a degree of robustness has its price. Instead of the volatility, the bond prices depend on the additional factor $q$, which is determined by the quadratic variation of the driving risk factor. That means, the term structure models we construct do not require any knowledge about how the volatility evolves in the future. All necessary information is included in the quadratic variation of the driver, that is, in the historical volatility. From a theoretical point of view, especially regarding the motivation of volatility uncertainty given in the introduction, it is desirable to have a term structure model that does not impose any assumptions on the future evolution of the volatility. The need to specify the evaluations of the process $q$ instead is not a problem, since this, in principle, can be done by inferring the evolution of the historical volatility from data. However, how these models perform in practice is a challenging question for future research, as many applications of term structure models involve estimation procedures and simulations \citep*{daisingleton2003}, which is a nontrivial task in the presence of a family of probability measures.

\section{Conclusion}\label{conclusion}
In this paper, we study the famous HJM model in the presence of volatility uncertainty. The main result is a sufficient condition, called drift condition, for the absence of arbitrage on the related bond market. In the presence of volatility uncertainty, the absence of arbitrage requires additional market prices, which are referred to as the market prices of uncertainty. The drift condition fully characterizes the risk-neutral dynamics of the forward rate in terms of its diffusion term. Since the latter also includes the uncertain volatility, the risk-neutral forward rate dynamics exhibit drift uncertainty. Using the drift condition, it is possible to construct arbitrage-free term structure models in the presence of volatility uncertainty, which we demonstrate in examples. In particular, we obtain robust versions of the Ho-Lee term structure and the Hull-White term structure, respectively. In examples where this is not possible, the drift condition shows how to adjust the model in order to be arbitrage-free when the volatility is uncertain, which we show in an example corresponding to the Vasicek term structure. The resulting term structures do not rely on any assumption how the volatility evolves in the future. Instead, the term structure is determined by the historical volatility. As a consequence, the resulting term structure models are completely robust with respect to the volatility.

\appendix

\section*{Appendix}

\section{Admissible Integrands for the Forward Rate}\label{admissible integrands for the forward rate}
We construct the space of admissible integrands for the forward rate dynamics as follows. Let us consider the measure space $([0,T],\mathcal{B}([0,T]),\lambda)$, where $T<\infty$, $\mathcal{B}([0,T])$ denotes the Borel $\sigma$-algebra on $[0,T]$, and $\lambda$ is the Lebesgue measure on $[0,T]$, and the space $M_G^p(0,T)$ for $p\geq1$, which is a Banach space with respect to the norm $\Vert\cdot\Vert_p$, defined by
\begin{align*}
\Vert\eta\Vert_p:=\hat{\mathbb{E}}\Big[\int_0^T\vert\eta_t\vert^pdt\Big]^\frac{1}{p}.
\end{align*}
for a process $\eta=(\eta_t)_{0\leq t\leq T}$ in $M_G^p(0,T)$. Then we define by $\tilde{M}_G^{p,0}(0,T)$ the space of all functions $\phi$ mapping from $[0,T]$ into $M_G^p(0,T)$ such that
\begin{align*}
\phi(s)=\sum_i\varphi^i1_{A_i}(s),
\end{align*}
where $(\varphi^i)_i$ is a finite sequence of processes $\varphi^i=(\varphi_t^i)_{0\leq t\leq T}$ in $M_G^p(0,T)$ and $(A_i)_i$ is a finite sequence of pairwise disjoint sets $A_i$ in $\mathcal{B}([0,T])$. We define the seminorm $\Vert\cdot\Vert_{\sim,p}$ by
\begin{align*}
\Vert\phi\Vert_{\sim,p}:=\int_0^T\Vert\phi(s)\Vert_pds
\end{align*}
on the space $\tilde{M}_G^{p,0}(0,T)$. By considering the quotient space with respect to the null space
\begin{align*}
\bar{M}_G^p(0,T):=\big\{\phi\in\tilde{M}_G^{p,0}(0,T)\big\vert\Vert\phi\Vert_{\sim,p}=0\big\},
\end{align*}
still denoted by $\tilde{M}_G^{p,0}(0,T)$, we get a normed space. The completion of $\tilde{M}_G^{p,0}(0,T)$ under the norm $\Vert\cdot\Vert_{\sim,p}$ is denoted by $\tilde{M}_G^p(0,T)$, being the space of admissible integrands.
\par There is an explicit representation of the space of admissible integrands. It can be shown \citep*[Section A.1]{prevotrockner2007} that the abstract completion of $\tilde{M}_G^{p,0}(0,T)$ is given by
\begin{align*}
\tilde{M}_G^p(0,T)=\big\{\phi:[0,T]\rightarrow M_G^p(0,T)\big\vert\phi\text{ is strongly measurable},\Vert\phi\Vert_{\sim,p}<\infty\big\}.
\end{align*}
A function $\phi:[0,T]\rightarrow M_G^p(0,T)$ is called \textit{strongly measurable} if it is $\mathcal{B}([0,T])/\mathcal{B}(M_G^p(0,T))$-measurable and $\phi([0,T])$ is separable. Therefore, we know that $\phi\in\tilde{M}_G^p(0,T)$ is a regular stochastic process for a fixed $s$, i.e., $\phi(s)\in M_G^p(0,T)$ for each $s$.
\par As a consequence, we can give a sufficient condition for functions to lie in this space.
\begin{prp}\label{sufficient condition}
Let $\phi:[0,T]\rightarrow M_G^p(0,T)$ be continuous. Then $\phi\in\tilde{M}_G^p(0,T)$.
\end{prp}
\begin{proof}
First, we show that $\phi$ is strongly measurable. Since $\phi$ is continuous, it is clearly $\mathcal{B}([0,T])/\mathcal{B}(M_G^p(0,T))$-measurable. Furthermore, $[0,T]$ is separable and the image of a continuous function with separable domain is separable. Hence, $\phi([0,T])$ is separable.
\par It is left to show that the norm of $\phi$ is finite. The norm $\Vert\cdot\Vert_p$ is obviously a continuous function. Thus, the function $f:\mathbb{R}\rightarrow\mathbb{R},s\mapsto\Vert\phi(s)\Vert_p$ is continuous, since $\phi$ is continuous. Therefore, $\Vert\phi\Vert_{\sim,p}<\infty$.
\end{proof}
\noindent By Proposition \ref{sufficient condition}, we have the following examples of functions in $\tilde{M}_G^p(0,T)$. First, we know that continuous real-valued functions on $[0,T]\times[0,T]$ belong to $\tilde{M}_G^p(0,T)$.
\begin{ex}\label{continuous functions are in M tilde}
The function $\phi:[0,T]\rightarrow M_G^p(0,T),s\mapsto f(\cdot,s)$ belongs to $\tilde{M}_G^p(0,T)$, where $f:[0,T]\times[0,T]\rightarrow\mathbb{R}$ is continuous. $\phi$ maps into $M_G^p(0,T)$, since $f(\cdot,s)$ is a continuous function for all $s$. This can be deduced from the representation of the space $M_G^p(0,T)$ \citep*[Theorem 4.7]{huwangzheng2016}. The continuity of $\phi$ follows from
\begin{align*}
\Vert\phi(s)-\phi(\tilde{s})\Vert_p=\Big(\int_0^T\vert f(t,s)-f(t,\tilde{s})\vert^pdt\Big)^\frac{1}{p}
\end{align*}
and dominated convergence. Thus, by Proposition \ref{sufficient condition}, $\phi\in\tilde{M}_G^p(0,T)$.
\end{ex}
\noindent Second, the product of a continuous real-valued function on $[0,T]$ and an admissible stochastic process lies in the space $\tilde{M}_G^p(0,T)$.
\begin{ex}\label{products of continuous functions and stochastic processes are in m tilde}
The function $\phi:[0,T]\rightarrow M_G^p(0,T),s\mapsto f(s)\eta$ belongs to $\tilde{M}_G^p(0,T)$, where $f:[0,T]\rightarrow\mathbb{R}$ is continuous and $\eta=(\eta_t)_{0\leq t\leq T}$ belongs to $M_G^p(0,T)$. Clearly, $\phi$ maps into $M_G^p(0,T)$. The continuity of $\phi$ follows from
\begin{align*}
\Vert\phi(s)-\phi(\tilde{s})\Vert_p=\Vert\eta\Vert_p\vert f(s)-f(\tilde{s})\vert.
\end{align*}
Hence, Proposition \ref{sufficient condition} implies that $\phi\in\tilde{M}_G^p(0,T)$.
\end{ex}
\par We are able to define integrals and, more importantly, double integrals for functions in $\tilde{M}_G^p(0,T)$. First, we define the Bochner integral for simple functions $\phi\in\tilde{M}_G^{p,0}(0,T)$,
\begin{align*}
\int_0^T\phi(s)ds:=\sum_i\varphi^i\lambda(A_i).
\end{align*}
The Bochner integral is a linear operator mapping from $\tilde{M}_G^{p,0}(0,T)$ into $M_G^p(0,T)$. In addition, the operator is continuous, since we have the inequality
\begin{align}\label{integral inequality}
\Big\Vert\int_0^T\phi(s)ds\Big\Vert_p\leq\int_0^T\Vert\phi(s)\Vert_pds.
\end{align}
Thus, we can extend the operator to the completion $\tilde{M}_G^p(0,T)$, still satisfying \eqref{integral inequality}. For $A\in\mathcal{B}([0,T])$, we define $\int_A\phi(s)ds:=\int_0^T1_A(s)\phi(s)ds$. Since the integral maps into $M_G^p(0,T)$, we can define the double integral $\int_0^T\int_A\phi_t(s)dsdB_t^i$ for $\phi\in\tilde{M}_G^2(0,T)$, mapping into $L_G^2(\Omega_T)$, and the double integrals $\int_0^T\int_A\psi_t(s)dsdt$ and $\int_0^T\int_A\psi_t(s)dsd\langle B^i,B^j\rangle_t$ for $\psi\in\tilde{M}_G^1(0,T)$, mapping into $L_G^1(\Omega_T)$, for all $i,j=1,...,d$.
\par We can also define double integrals for the reversed order of integration and interchange the order. For this purpose, we use the following result \citep*[Proposition A.2.2]{prevotrockner2007}.
\begin{prp}\label{result from bochner integration}
Let $\phi\in\tilde{M}_G^p(0,T)$, $X$ be a Banach space, and $F:M_G^p(0,T)\rightarrow X$ be a continuous linear operator. Then we can define $\int_AF(\phi(s))ds$, mapping into $X$, and it holds
\begin{align*}
\int_A(F\circ\phi)(s)ds=F\Big(\int_A\phi(s)ds\Big).
\end{align*}
\end{prp}
\noindent All stochastic integrals related to the $G$-Brownian motion are continuous linear operators. Thus, Proposition \ref{result from bochner integration} allows us to define the integral $\int_A\int_0^T\phi_t(s)dB_t^ids$ for $\phi\in\tilde{M}_G^2(0,T)$, mapping into $L_G^2(\Omega_T)$, and the integrals $\int_A\int_0^T\psi_t(s)dtds$ and $\int_A\int_0^T\psi_t(s)d\langle B^i,B^j\rangle_tds$ for $\psi\in\tilde{M}_G^1(0,T)$, mapping into $L_G^1(\Omega_T)$, for all $i,j=1,...,d$. Moreover, we obtain a version of Fubini's theorem, which is an essential tool in the HJM model.
\begin{cor}\label{fubini}
Let $\phi\in\tilde{M}_G^2(0,T)$ and $\psi\in\tilde{M}_G^1(0,T)$. Then, for all $i,j$, it holds
\begin{align*}
\int_A\int_0^T\phi_t(s)dB_t^ids={}&\int_0^T\int_A\phi_t(s)dsdB_t^i,
\\\int_A\int_0^T\psi_t(s)dtds={}&\int_0^T\int_A\psi_t(s)dsdt,
\\\int_A\int_0^T\psi_t(s)d\langle B^i,B^j\rangle_tds={}&\int_0^T\int_A\psi_t(s)dsd\langle B^i,B^j\rangle_t.
\end{align*}
\end{cor}
\par In addition to the double integrals from above, we need to define more complex integrals. In order to achieve this, we need the following proposition.
\begin{prp}\label{complex integrals}
Let $\phi\in\tilde{M}_G^p(0,T)$ and $\psi:[0,T]\rightarrow M_G^p(0,T),s\mapsto1_{[0,s]}\phi(s)$. Then we have $\psi\in\tilde{M}_G^p(0,T)$.
\end{prp}
\begin{proof}
First, we decompose $\psi$ into several functions to show that it is strongly measurable. Let us define the subspace
\begin{align*}
\mathcal{M}^p:=\big\{\eta\in M_G^p(0,T)\big\vert\eta=1_{[0,a]},a\in[0,T]\big\}\subset M_G^p(0,T).
\end{align*}
and the functions $f:[0,T]\rightarrow\mathcal{M}^p,s\mapsto1_{[0,s]}$, $g:[0,T]\rightarrow\mathcal{M}^p\times M_G^p(0,T),s\mapsto(f(s),\phi(s))$, and $h:\mathcal{M}^p\times M_G^p(0,T)\rightarrow M_G^p(0,T),(\eta,\zeta)\mapsto\eta\zeta$. Then we have $\psi=h\circ g$.
\par We deduce the measurability of $\psi$ from the measurability of the decomposition. $f$ is continuous, since
\begin{align*}
\Vert1_{[0,a]}-1_{[0,\tilde{a}]}\Vert_p=\vert a-\tilde{a}\vert
\end{align*}
for $a,\tilde{a}\in[0,T]$, and thus, $f$ is $\mathcal{B}([0,T])/\mathcal{B}(\mathcal{M}^p)$-measurable. By assumption, $\phi$ is $\mathcal{B}([0,T])/\mathcal{B}(M_G^p(0,T))$-measurable and so $g$ is $\mathcal{B}([0,T])/\mathcal{B}(\mathcal{M}^p)\otimes\mathcal{B}(M_G^p(0,T))$-measurable. Now it is left to show that $h$ is $\mathcal{B}(\mathcal{M}^p)\otimes\mathcal{B}(M_G^p(0,T))/\mathcal{B}(M_G^p(0,T))$-measurable to deduce the measurability of $\psi$. We equip $\mathcal{M}^p\times M_G^p(0,T)$ with the norm $\Vert\cdot\Vert$, defined by
\begin{align*}
\Vert(\eta,\zeta)\Vert:=\max\{\Vert\eta\Vert_p,\Vert\zeta\Vert_p\}.
\end{align*}
For $a,\tilde{a}\in[0,T]$ and $\zeta=(\zeta_t)_{0\leq t\leq T}$ and $\tilde{\zeta}=(\tilde{\zeta}_t)_{0\leq t\leq T}$ in $M_G^p(0,T)$, we have
\begin{align*}
\Vert1_{[0,a]}\zeta-1_{[0,\tilde{a}]}\tilde{\zeta}\Vert_p\leq\Vert\zeta-\tilde{\zeta}\Vert_p+\Vert1_{[a,\tilde{a}]}\tilde{\zeta}\Vert_p,
\end{align*}
where the last term converges to $0$ as $a$ converges to $\tilde{a}$. Therefore, $h$ is continuous. So we need to show that $M\in\mathcal{B}(\mathcal{M}^p)\otimes\mathcal{B}(M_G^p(0,T))$ for an open set $M\subset\mathcal{M}^p\times M_G^p(0,T)$ to obtain the measurability of $h$. If $M$ is open, we can represent it as a union of open balls in $\mathcal{M}^p\times M_G^p(0,T)$. By the definition of $\Vert\cdot\Vert$, every open ball in $\mathcal{M}^p\times M_G^p(0,T)$ can be written as a rectangle whose sides are open balls in $\mathcal{M}^p$ and $M_G^p(0,T)$, respectively. Hence, we can represent $M$ as a rectangle whose sides are open sets in $\mathcal{M}^p$ and $M_G^p(0,T)$, respectively. Therefore, $M\in\mathcal{B}(\mathcal{M}^p)\otimes\mathcal{B}(M_G^p(0,T))$.
\par We show that the image of $\psi$ is separable by using the continuity from the second step. By assumption, the image of $\phi$ is separable. Moreover, $[0,T]$ is separable and $f$ is continuous. Thus, $f([0,T])$ is separable, implying that the image of $g$ is separable. Due to the continuity of $h$, also $h(g([0,T]))$ is separable. Since $\psi=h\circ g$, it follows that $\psi([0,T])$ is separable.
\par It is left to show that the norm of $\psi$ is finite. We have
\begin{align*}
\int_0^T\Vert\psi(s)\Vert_pds\leq\int_0^T\Vert\phi(s)\Vert_pds<\infty,
\end{align*}
which completes the proof.
\end{proof}
\noindent Due to Proposition \ref{complex integrals}, we are able to define integrals of the form $\int_0^T\int_t^T\phi_t(s)dsdB_t^i$ and, by Proposition \ref{result from bochner integration}, $\int_0^T\int_0^s\phi_t(s)dB_t^ids$ for $\phi\in\tilde{M}_G^2(0,T)$, mapping into $L_G^2(\Omega_T)$, for all $i$. Moreover, Corollary \ref{fubini} implies that, for all $i$,
\begin{align*}
\int_0^T\int_t^T\phi_t(s)dsdB_t^i=\int_0^T\int_0^s\phi_t(s)dB_t^ids.
\end{align*}
The same holds if we replace $dB_t^i$, $\tilde{M}_G^2(0,T)$, and $L_G^2(\Omega_T)$ by $dt$, as well as $d\langle B^i,B^j\rangle_t$ for all $j$, $\tilde{M}_G^1(0,T)$, and $L_G^1(\Omega_T)$, respectively.
\par In the end, we deal with the differentiability of integrals and, especially, the differentiability of products of two integrals, used for the calculations on the diffusion coefficient of the forward rate. As the classical Lebesgue integral, integrals of functions in $\tilde{M}_G^p(0,T)$ are, loosely speaking, differentiable and absolutely continuous.
\begin{prp}\label{differentiability and absolute continuity of the integral}
Let $\phi\in\tilde{M}_G^p(0,T)$ and $\Phi:[0,T]\rightarrow M_G^p(0,T),s\mapsto\int_0^s\phi(u)du$. Then
\begin{enumerate}
\item[(i)] $\Phi$ is almost everywhere differentiable and $\Phi'=\phi$,
\item[(ii)] for all $\epsilon>0$, there exists a $\delta>0$ such that
\begin{align*}
\sum_i\Vert\Phi(s_i)-\Phi(\tilde{s}_i)\Vert_p<\epsilon
\end{align*}
for every finite sequence of disjoint open intervals $((\tilde{s}_i,s_i))_i$ such that $\sum_i(s_i-\tilde{s}_i)<\delta$.
\end{enumerate}
\end{prp}
\begin{proof}
We deduce all statements from the inequality \eqref{integral inequality} and the properties of the Lebesgue integral. By \eqref{integral inequality}, we have
\begin{align*}
\big\Vert\tfrac{1}{s-\tilde{s}}\big(\Phi(s)-\Phi(\tilde{s})\big)-\phi(\tilde{s})\big\Vert_p\leq\big\vert\tfrac{1}{s-\tilde{s}}\big(f(s)-f(\tilde{s})\big)\big\vert,
\end{align*}
where $f:[0,T]\rightarrow\mathbb{R}$ is defined by $f(s):=\int_0^s\Vert\phi(u)-\phi(\tilde{s})\Vert_pdu$. Due to the differentiability of the Lebesgue integral, the expression on the right-hand side of the previous inequality converges to $0$ as $s$ converges to $\tilde{s}$ for almost all $\tilde{s}$. Therefore, $\Phi$ is almost everywhere differentiable and $\Phi'=\phi$. Furthermore, we can use \eqref{integral inequality} to obtain
\begin{align*}
\Vert\Phi(s)-\Phi(\tilde{s})\Vert_p\leq\vert g(s)-g(\tilde{s})\vert,
\end{align*}
where $g:[0,T]\rightarrow\mathbb{R}$ is defined by $g(s):=\int_0^s\Vert\phi(u)\Vert_pdu$. Since the Lebesgue integral is absolutely continuous, for all $\epsilon>0$, we can find a $\delta>0$ such that
\begin{align*}
\sum_i\Vert\Phi(s_i)-\Phi(\tilde{s}_i)\Vert_p\leq\sum_i\vert g(s_i)-g(\tilde{s}_i)\vert<\epsilon
\end{align*}
for every finite sequence of disjoint open intervals $((\tilde{s}_i,s_i))_i$ such that $\sum_i(s_i-\tilde{s}_i)<\delta$.
\end{proof}
\noindent Conversely, we have a version of the fundamental theorem of calculus for functions in $\tilde{M}_G^p(0,T)$ which are differentiable and absolutely continuous.
\begin{prp}\label{fundamental theorem of calculus}
Let $\Phi:[0,T]\rightarrow M_G^p(0,T)$ be almost everywhere differentiable and $\Phi'=\phi$, where $\phi\in\tilde{M}_G^p(0,T)$. If for all $\epsilon>0$, there exists a $\delta>0$ such that
\begin{align*}
\sum_i\Vert\Phi(s_i)-\Phi(\tilde{s}_i)\Vert_p<\epsilon
\end{align*}
for every finite sequence of disjoint open intervals $((\tilde{s}_i,s_i))_i$ such that $\sum_i(s_i-\tilde{s}_i)<\delta$, then it holds
\begin{align*}
\Phi(s)-\Phi(0)=\int_0^s\phi(u)du.
\end{align*}
\end{prp}
\begin{proof}
We prove the assertion by using a consequence of the Hahn-Banach theorem and the fundamental theorem of calculus for Lebesgue integrals. Let $F:M_G^p(0,T)\rightarrow\mathbb{R}$ be a continuous linear functional. Then we have
\begin{align*}
\vert(F\circ\Phi)(s)-(F\circ\Phi)(\tilde{s})\vert\leq C\Vert\Phi(s)-\Phi(\tilde{s})\Vert_p
\end{align*}
for some constant $C>0$. Due to the last assumption on $\Phi$, we deduce that $F\circ\Phi$ is absolutely continuous. The fundamental theorem of calculus for Lebesgue integrals implies that $F\circ\Phi$ is almost everywhere differentiable, its derivative is integrable, and
\begin{align*}
(F\circ\Phi)(s)-(F\circ\Phi)(0)=\int_0^s(F\circ\Phi)'(u)du.
\end{align*}
Furthermore, $(F\circ\Phi)'=F\circ\phi$, since the continuity and the linearity of $F$ imply
\begin{align*}
\vert(F\circ\Phi)'(\tilde{s})-(F\circ\phi)(\tilde{s})\vert\leq{}&\big\vert(F\circ\Phi)'(\tilde{s})-\tfrac{1}{s-\tilde{s}}\big((F\circ\Phi)(s)-(F\circ\Phi)(\tilde{s})\big)\big\vert
\\&{}+C\big\Vert\tfrac{1}{s-\tilde{s}}\big(\Phi(s)-\Phi(\tilde{s})\big)-\phi(\tilde{s})\big\Vert_p
\end{align*}
for some constant $C>0$, where the terms on the right-hand side converge to $0$ as $s$ converges to $\tilde{s}$ for almost all $\tilde{s}$. Hence, we obtain
\begin{align*}
(F\circ\Phi)(s)-(F\circ\Phi)(0)=\int_0^s(F\circ\phi)(u)du.
\end{align*}
By Proposition \ref{result from bochner integration} and the linearity of $F$, it holds
\begin{align*}
F\Big(\Phi(s)-\Phi(0)-\int_0^s\phi(u)du\Big)=0.
\end{align*}
Since this holds for every continuous linear functional, the Hahn-Banach theorem implies the assertion.
\end{proof}
\noindent In order to derive a product rule for differentiable and absolutely continuous functions, we use the following lemma.
\begin{lem}\label{products of processes in m tilde are in m tilde}
Let $\phi,\psi\in\tilde{M}_G^{2p}(0,T)$ and $\upsilon:[0,T]\rightarrow M_G^p(0,T),s\mapsto\phi(s)\psi(s)$, where $\phi$ is continuous. Then $\upsilon\in\tilde{M}_G^p(0,T)$.
\end{lem}
\begin{proof}
First of all, we know that $\upsilon$ maps into $M_G^p(0,T)$, since one can show, by H\"older's inequality, that the product of two processes in $M_G^{2p}(0,T)$ belongs to $M_G^p(0,T)$.
\par Next, we show that $\upsilon$ is strongly measurable. The strong measurability follows as in the proof of Proposition \ref{complex integrals} if $f:M_G^{2p}(0,T)\times M_G^{2p}(0,T)\rightarrow M_G^p(0,T)$, defined by $f(\eta,\zeta):=\eta\zeta$, is continuous, where we equip $M_G^{2p}(0,T)\times M_G^{2p}(0,T)$ with the norm $\Vert\cdot\Vert$, defined by
\begin{align*}
\Vert(\eta,\zeta)\Vert:=\max\{\Vert\eta\Vert_{2p},\Vert\zeta\Vert_{2p}\}.
\end{align*}
By applying H\"older's inequality, for $(\eta,\zeta)=(\eta_t,\zeta_t)_{0\leq t\leq T}$ and $(\tilde{\eta},\tilde{\zeta})=(\tilde{\eta}_t,\tilde{\zeta}_t)_{0\leq t\leq T}$ in $M_G^{2p}(0,T)\times M_G^{2p}(0,T)$, we have
\begin{align*}
\Vert\eta\zeta-\tilde{\eta}\tilde{\zeta}\Vert_p\leq\Vert\eta\Vert_{2p}\Vert\zeta-\tilde{\zeta}\Vert_{2p}+\Vert\eta-\tilde{\eta}\Vert_{2p}\Vert\tilde{\zeta}\Vert_{2p}.
\end{align*}
Therefore, $f$ is continuous, since $\Vert\eta\Vert_{2p}$ and $\Vert\tilde{\zeta}\Vert_{2p}$ are finite and $\Vert\eta-\tilde{\eta}\Vert_{2p}$ and $\Vert\zeta-\tilde{\zeta}\Vert_{2p}$ converge to $0$ as $(\eta,\zeta)$ converges to $(\tilde{\eta},\tilde{\zeta})$.
\par Finally, we deduce that the norm of $\upsilon$ is finite. The function $s\mapsto\Vert\phi(s)\Vert_{2p}$ is continuous due to the continuity of $\phi$ and the norm and hence, bounded. Thus,
\begin{align*}
\int_0^T\Vert\upsilon(s)\Vert_pds\leq\int_0^T\Vert\phi(s)\Vert_{2p}\Vert\psi(s)\Vert_{2p}ds\leq C\int_0^T\Vert\psi(s)\Vert_{2p}ds<\infty
\end{align*}
for some constant $C>0$, which completes the proof.
\end{proof}
\noindent With the help of Lemma \ref{products of processes in m tilde are in m tilde}, we obtain the desired product rule.
\begin{prp}\label{product rule}
Let $\Phi,\Psi:[0,T]\rightarrow M_G^{2p}(0,T)$ be almost everywhere differentiable and $\Phi'=\phi$ and $\Psi'=\psi$, where $\phi,\psi\in\tilde{M}_G^{2p}(0,T)$. If for all $\epsilon>0$, there exists a $\delta>0$ such that
\begin{align*}
\sum_i\Vert\Phi(s_i)-\Phi(\tilde{s}_i)\Vert_p,\sum_i\Vert\Psi(s_i)-\Psi(\tilde{s}_i)\Vert_p<\epsilon
\end{align*}
for every finite sequence of disjoint open intervals $((\tilde{s}_i,s_i))_i$ such that $\sum_i(s_i-\tilde{s}_i)<\delta$, then it holds
\begin{align*}
\Phi(s)\Psi(s)-\Phi(0)\Psi(0)=\int_0^s\phi(u)\Psi(u)+\Phi(u)\psi(u)du.
\end{align*}
\end{prp}
\begin{proof}
We show that $\Upsilon:[0,T]\rightarrow M_G^p(0,T),s\mapsto\Phi(s)\Psi(s)$ satisfies the assumptions of Proposition \ref{fundamental theorem of calculus} to deduce the assertion. Using H\"older's inequality, we have
\begin{align*}
\Vert\Upsilon(s)-\Upsilon(\tilde{s})\Vert_p\leq\Vert\Phi(s)\Vert_{2p}\Vert\Psi(s)-\Psi(\tilde{s})\Vert_{2p}+\Vert\Phi(s)-\Phi(\tilde{s})\Vert_{2p}\Vert\Psi(\tilde{s})\Vert_{2p}.
\end{align*}
The last assumption implies that $s\mapsto\Vert\Phi(s)\Vert_{2p}$ and $\tilde{s}\mapsto\Vert\Psi(\tilde{s})\Vert_{2p}$ are continuous and thus, bounded. Hence, for all $\epsilon>0$, we can find a $\delta>0$ such that
\begin{align*}
\sum_i\Vert\Upsilon(s_i)-\Upsilon(\tilde{s}_i)\Vert_p<\epsilon
\end{align*}
for every finite sequence of disjoint open intervals $((\tilde{s}_i,s_i))_i$ such that $\sum_i(s_i-\tilde{s}_i)<\delta$. Let $\upsilon:[0,T]\rightarrow M_G^p(0,T),s\mapsto\phi(s)\Psi(s)+\Phi(s)\psi(s)$. By Lemma \ref{products of processes in m tilde are in m tilde}, $\upsilon\in\tilde{M}_G^p(0,T)$, since $\phi,\psi\in\tilde{M}_G^{2p}(0,T)$ and $\Phi$ and $\Psi$ are continuous and belong to $\tilde{M}_G^{2p}(0,T)$ by the assumptions. Furthermore, we have
\begin{align*}
\big\Vert\tfrac{1}{s-\tilde{s}}\big(\Upsilon(s)-\Upsilon(\tilde{s})\big)-\upsilon(\tilde{s})\big\Vert_p\leq{}&\big\Vert\tfrac{1}{s-\tilde{s}}\big(\Upsilon(s)-\Phi(\tilde{s})\Psi(s)\big)-\phi(\tilde{s})\Psi(s)\big\Vert_p
\\&{}+\Vert\phi(\tilde{s})\Psi(s)-\phi(\tilde{s})\Psi(\tilde{s})\Vert_p
\\&{}+\big\Vert\tfrac{1}{s-\tilde{s}}\big(\Phi(\tilde{s})\Psi(s)-\Upsilon(\tilde{s})\big)-\Phi(\tilde{s})\psi(\tilde{s})\big\Vert_p.
\end{align*}
The three terms on the right-hand side converge to $0$ as $s$ converges to $\tilde{s}$ for almost all $\tilde{s}$. The first term converges almost everywhere to $0$, since
\begin{align*}
\big\Vert\tfrac{1}{s-\tilde{s}}\big(\Upsilon(s)-\Phi(\tilde{s})\Psi(s)\big)-\phi(\tilde{s})\Psi(s)\big\Vert_p\leq\big\Vert\tfrac{1}{s-\tilde{s}}\big(\Phi(s)-\Phi(\tilde{s})\big)-\phi(\tilde{s})\big\Vert_{2p}\Vert\Psi(s)\Vert_{2p},
\end{align*}
$\Phi$ is almost everywhere differentiable, $\Phi'=\phi$, and $s\mapsto\Vert\Psi(s)\Vert_{2p}$ is continuous. The second term converges to $0$, since
\begin{align*}
\Vert\phi(\tilde{s})\Psi(s)-\phi(\tilde{s})\Psi(\tilde{s})\Vert_p\leq\Vert\phi(\tilde{s})\Vert_{2p}\Vert\Psi(s)-\Psi(\tilde{s})\Vert_{2p},
\end{align*}
$\phi(\tilde{s})\in M_G^{2p}(0,T)$, and $s\mapsto\Vert\Psi(s)\Vert_{2p}$ is continuous. The third term converges to $0$ almost everywhere, since
\begin{align*}
\big\Vert\tfrac{1}{s-\tilde{s}}\big(\Phi(\tilde{s})\Psi(s)-\Upsilon(\tilde{s})\big)-\Phi(\tilde{s})\psi(\tilde{s})\big\Vert_{p}\leq\Vert\Phi(\tilde{s})\Vert_{2p}\big\Vert\tfrac{1}{s-\tilde{s}}\big(\Psi(s)-\Psi(\tilde{s})\big)-\psi(\tilde{s})\big\Vert_{2p},
\end{align*}
$\Phi(\tilde{s})\in M_G^{2p}(0,T)$, $\Psi$ is almost everywhere differentiable, and $\Psi'=\psi$. Therefore, $\Upsilon$ is almost everywhere differentiable and $\Upsilon'=\upsilon$. This completes the proof.
\end{proof}
\noindent Combining Proposition \ref{differentiability and absolute continuity of the integral} and Proposition \ref{product rule}, we have the following result, which we apply to the diffusion coefficient of the forward rate.
\begin{cor}\label{product rule for the diffusion coefficient}
Let $\phi,\psi\in\tilde{M}_G^{2p}(0,T)$ and $\Phi,\Psi:[0,T]\rightarrow M_G^p(0,T)$ be defined by $\Phi(s):=\int_0^s\phi(u)du$ and $\Psi(s):=\int_0^s\psi(u)du$, respectively. Then it holds
\begin{align*}
\Phi(s)\Psi(s)=\int_0^s\phi(u)\Psi(u)+\Phi(u)\psi(u)du.
\end{align*}
\end{cor}

\section{Regularity of the Discounted Bonds}\label{regularity of the discounted bonds}
In order to show that the discounted bonds are sufficiently regular, we consider the exponential of a diffusion process driven by a $G$-Brownian motion. Let us define the process $X=(X_t)_{0\leq t\leq T}$ by
\begin{align*}
X_t:=\exp\Big(\int_0^ta_udu+\sum_{i=1}^d\int_0^tb_u^idB_u^i+\sum_{i,j=1}^d\int_0^tc_u^{i,j}d\langle B^i,B^j\rangle_u\Big),
\end{align*}
where $a=(a_t)_{0\leq t\leq T}$ and $c^{i,j}=(c_t^{i,j})_{0\leq t\leq T}$ belong to $M_G^1(0,T)$ and $b^i=(b_t^i)_{0\leq t\leq T}$ belongs to $M_G^2(0,T)$ for all $i,j=1,...,d$. Then the dynamics of $X$ are given by
\begin{align*}
X_t=1+\int_0^ta_uX_udu+\sum_{i=1}^d\int_0^tb_u^iX_udB_u^i+\sum_{i,j=1}^d\int_0^t(c_u^{i,j}+\tfrac{1}{2}b_u^ib_u^j)X_ud\langle B^i,B^j\rangle_u
\end{align*}
by It\^o's formula for $G$-Brownian motion.
\par The following result provides a sufficient condition ensuring that the dynamics of the process $X$ are regular. As a consequence, we obtain that $X$ itself is well-posed.
\begin{prp}\label{dynamics of the discounted bond are regular}
If there exists a $p>1$ such that $a,c^{i,j}\in M_G^p(0,T)$ and $b^i\in M_G^{2p}(0,T)$ for all $i,j=1,...,d$ and there exist $\tilde{p}>p^*$ and $\tilde{q}>2$, where $p^*:=\frac{2pq}{p-q}$ for some $q\in(1,p)$ such that
\begin{align*}
\hat{\mathbb{E}}\Big[\int_0^T\exp\Big(\tfrac{\tilde{p}\tilde{q}}{\tilde{q}-2}\Big(\int_0^ta_udu+\sum_{i,j=1}^d\int_0^tc_u^{i,j}d\langle B^i,B^j\rangle_u\Big)\Big)dt\Big]<{}&\infty,
\\\hat{\mathbb{E}}\Big[\int_0^T\exp\Big(\tfrac{1}{2}(\tilde{p}\tilde{q})^2\sum_{i,j=1}^d\int_0^tb_u^ib_u^jd\langle B^i,B^j\rangle_u\Big)dt\Big]<{}&\infty,
\end{align*}
then we have $X\in M_G^{p^*}(0,T)$. In particular, this implies $aX,(c^{i,j}+\frac{1}{2}b^ib^j)X\in M_G^1(0,T)$ and $b^iX\in M_G^2(0,T)$ for all $i,j=1,...,d$.
\end{prp}
\begin{proof}
In order to show that $X\in M_G^{p^*}(0,T)$, we use the characterization of the space $M_G^{p^*}(0,T)$ from \citet*{huwangzheng2016}. The space $M_G^{p^*}(0,T)$ consists of all progressively measurable processes $\eta=(\eta_t)_{0\leq t\leq T}$ such that $\Vert\eta\Vert_{p^*}^{p^*}<\infty$, $\eta$ has a quasi-continuous version, and
\begin{align*}
\lim_{n\rightarrow\infty}\hat{\mathbb{E}}\Big[\int_0^T\vert\eta_t\vert^{p^*}1_{\{\vert\eta_t\vert\geq n\}}dt\Big]=0
\end{align*}
\citep*[Theorem 4.7]{huwangzheng2016}. Since $a,c^{i,j}\in M_G^p(0,T)$ and $b^i\in M_G^{2p}(0,T)$ for all $i,j$, we know that $X$ is progressively measurable and has a quasi-continuous version. Therefore, since
\begin{align*}
\hat{\mathbb{E}}\Big[\int_0^T\vert X_t\vert^{p^*}1_{\{\vert X_t\vert\geq n\}}dt\Big]\leq\tfrac{1}{n^{\tilde{p}-p^*}}\hat{\mathbb{E}}\Big[\int_0^T\vert X_t\vert^{\tilde{p}}dt\Big],
\end{align*}
it is left to show that $\Vert X\Vert_{\tilde{p}}^{\tilde{p}}<\infty$ in order to deduce that $X\in M_G^{p^*}(0,T)$. We have
\begin{align*}
\hat{\mathbb{E}}\Big[\int_0^T\vert X_t\vert^{\tilde{p}}dt\Big]=\hat{\mathbb{E}}\Big[&\int_0^T\exp\Big(\tilde{p}\sum_{i=1}^d\int_0^tb_u^idB_u^i-\tfrac{1}{2}\tilde{p}^2\tilde{q}\sum_{i,j=1}^d\int_0^tb_u^ib_u^jd\langle B^i,B^j\rangle_u\Big)
\\&{}\times\exp\Big(\tfrac{1}{2}\tilde{p}^2\tilde{q}\sum_{i,j=1}^d\int_0^tb_u^ib_u^jd\langle B^i,B^j\rangle_u\Big)
\\&{}\times\exp\Big(\tilde{p}\Big(\int_0^ta_udu+\sum_{i,j=1}^d\int_0^tc_u^{i,j}d\langle B^i,B^j\rangle_u\Big)\Big)dt\Big].
\end{align*}
By H\"older's inequality, we get
\begin{align*}
\hat{\mathbb{E}}\Big[\int_0^T\vert X_t\vert^{\tilde{p}}dt\Big]\leq{}&\hat{\mathbb{E}}\Big[\int_0^T\exp\Big(\tilde{p}\tilde{q}\sum_{i=1}^d\int_0^tb_u^idB_u^i-\tfrac{1}{2}(\tilde{p}\tilde{q})^2\sum_{i,j=1}^d\int_0^tb_u^ib_u^jd\langle B^i,B^j\rangle_u\Big)dt\Big]^\frac{1}{\tilde{q}}
\\&{}\times\hat{\mathbb{E}}\Big[\int_0^T\exp\Big(\tfrac{1}{2}(\tilde{p}\tilde{q})^2\sum_{i,j=1}^d\int_0^tb_u^ib_u^jd\langle B^i,B^j\rangle_u\Big)dt\Big]^\frac{1}{\tilde{q}}
\\&{}\times\hat{\mathbb{E}}\Big[\int_0^T\exp\Big(\tfrac{\tilde{p}\tilde{q}}{\tilde{q}-2}\Big(\int_0^ta_udu+\sum_{i,j=1}^d\int_0^tc_u^{i,j}d\langle B^i,B^j\rangle_u\Big)\Big)dt\Big]^\frac{\tilde{q}-2}{\tilde{q}}.
\end{align*}
By assumption, we know that the second and the third term on the right-hand side are finite. Hence, it is left to show that the first term is also finite. We can use the classical Fubini theorem and the last assumption to get
\begin{align*}
\mathbb{E}_P\Big[\exp\Big(\tfrac{1}{2}(\tilde{p}\tilde{q})^2\sum_{i,j=1}^d\int_0^tb_u^ib_u^jd\langle B^i,B^j\rangle_u\Big)\Big]<\infty
\end{align*}
for almost all $t$ for all $P\in\mathcal{P}$. Thus, we know that Novikov's condition is satisfied, which implies that
\begin{align*}
\mathbb{E}_P\Big[\exp\Big(\tilde{p}\tilde{q}\sum_{i=1}^d\int_0^tb_u^idB_u^i-\tfrac{1}{2}(\tilde{p}\tilde{q})^2\sum_{i,j=1}^d\int_0^tb_u^ib_u^jd\langle B^i,B^j\rangle_u\Big)\Big]=1
\end{align*}
for almost all $t$ for all $P\in\mathcal{P}$. Integrating and using Fubini's theorem once more, we obtain
\begin{align*}
\mathbb{E}_P\Big[\int_0^T\exp\Big(\tilde{p}\tilde{q}\sum_{i=1}^d\int_0^tb_u^idB_u^i-\tfrac{1}{2}(\tilde{p}\tilde{q})^2\sum_{i,j=1}^d\int_0^tb_u^ib_u^jd\langle B^i,B^j\rangle_u\Big)dt\Big]=T
\end{align*}
for all $P\in\mathcal{P}$, which implies the desired finiteness.
\par We are left to show that $aX,(c^{i,j}+\frac{1}{2}b^ib^j)X\in M_G^1(0,T)$ and $b^iX\in M_G^2(0,T)$ for all $i,j$. By the argument from the first step, we need to show that $\Vert aX\Vert_q^q<\infty$ in order to deduce that $aX\in M_G^1(0,T)$. By H\"older's inequality, it holds
\begin{align*}
\hat{\mathbb{E}}\Big[\int_0^T\vert a_tX_t\vert^qdt\Big]\leq\hat{\mathbb{E}}\Big[\int_0^T\vert a_t\vert^pdt\Big]^\frac{q}{p}\hat{\mathbb{E}}\Big[\int_0^T\vert X_t\vert^{\frac{1}{2}p^*}dt\Big]^\frac{p-q}{p}.
\end{align*}
The two terms on the right-hand side are finite, since $a\in M_G^p(0,T)$ and $X\in M_G^{p^*}(0,T)$. Thus, we obtain $\Vert aX\Vert_q^q<\infty$. We can show that $(c^{i,j}+\frac{1}{2}b^ib^j)X\in M_G^1(0,T)$ and $b^iX\in M_G^2(0,T)$ for all $i,j$ in the same way.
\end{proof}
\noindent If $a$, $b^i$, and $c^{i,j}$, for all $i,j$, satisfy the assumptions of Proposition \ref{dynamics of the discounted bond are regular}, we get $X_t\in L_G^1(\Omega_t)$ for all $t$ by It\^o's formula. If, in addition, $a=0=c^{i,j}+\frac{1}{2}b^ib^j$ for all $i,j$, we even have $X_t\in L_G^2(\Omega_t)$ for all $t$.

\bibliography{C:/Users/jhoelzermann/Documents/Uni/Literature/Literature}
\bibliographystyle{chicago}

\end{document}